\documentclass[envcountsect,envcountsame,11pt]{article}
\usepackage{fullpage}

\usepackage[usenames,dvipsnames]{pstricks}
\usepackage{pst-node}   
\usepackage{epsfig}

\usepackage{latexsym}
\usepackage{amsmath}
\usepackage{amssymb}
\usepackage{amsfonts}
\usepackage{ifthen}
\usepackage{revsymb}
\usepackage{yfonts}
\usepackage{graphicx}
\usepackage{multirow}
\usepackage{url}
\usepackage{xspace}
\usepackage{bbm}

\usepackage{hyperref}

\newcommand{\nsw}{noisy-storage model}

\newcommand{\scp}{strong-converse property}
\newcommand*{\cB}{\mathcal{B}}

\newcommand*{\cF}{\mathcal{F}}
\newcommand*{\cH}{\mathcal{H}}
\newcommand*{\cN}{\mathcal{N}}

\newcommand*{\sbin}{\{0,1\}}

\newcommand{\guess}{p_\mathrm{guess}}

\newcommand{\ket}[1]{\left\vert{#1}\right\rangle}

\newcommand{\cancel}[1]{}

\def\01{\{0,1\}}

\newcommand{\eps}{\varepsilon}

\newcommand{\sketbra}[2]{{\ensuremath{\lvert #1\rangle\!\langle #2\rvert}}}
\newcommand{\lketbra}[2]{{\ensuremath{\left\lvert #1\right\rangle\!\!\left\langle #2\right\rvert}}}
\newcommand{\ketbra}[2]{\if@display\lketbra{#1}{#2}\else\sketbra{#1}{#2}\fi}

\newcommand{\proj}[1]{\ketbra{#1}{#1}}

\newcommand\beq{\begin{equation}}
\newcommand\eeq{\end{equation}}
\newcommand\bea{\begin{eqnarray}}
\newcommand\eea{\end{eqnarray}}
\newcommand{\Tr}{\mbox{\rm Tr}}

\newtheorem{theorem}{Theorem}[section]
\newtheorem{lemma}[theorem]{Lemma}

\newtheorem{corollary}[theorem]{Corollary}
\newtheorem{protocol}{Protocol}
\newtheorem{definition}[theorem]{Definition}

\newenvironment{proof}
{\noindent {\bf Proof. }}
{{\hfill $\Box$}\\
 \smallskip}

\newcommand{\mY}{\mathcal{Y}}

\newcommand{\mX}{\mathcal{X}}

\newcommand{\mN}{\mathcal{N}}

\newcommand{\setI}{\mathcal{I}}

\newcommand{\setX}{\mathcal{X}}

\newcommand{\setF}{\mathcal{F}}

\newcommand{\set}[1]{\{#1\}}
\newcommand{\Set}[2]{\{ #1 : #2\}}

\newcommand{\regE}{E}

\newcommand{\assign}{\ensuremath{\kern.5ex\raisebox{.1ex}{\mbox{\rm:}}\kern -.3em =}}

\newcommand{\ol}[1]{\overline{#1}}

\renewcommand{\H}{\operatorname{H}} 
\newcommand{\hmin}{\ensuremath{\H_{\mathrm{min}}}}
\newcommand{\hmine}[2]{\ensuremath{\hmin^{#1}\left(#2\right)}}
\newcommand{\hminee}[1]{\hmine{\varepsilon}{#1}}

\newcommand{\syn}{\mathit{syn}} 

\newcommand{\dB}{{\sf B}'}

\newcommand{\ev}{\mathcal{E}}

\newcommand{\sent}{{\rm sent}}

\newcommand{\lostB}{{\rm B,no\ click}}

\newcommand{\errB}{{\rm B, err}}

\newcommand{\code}{\mathfrak{c}}
\newcommand{\QID}{{\sf \textit{Q-ID}}\xspace}

\newcommand{\regB}{\ensuremath{E}_{\dB}}

\begin{document}

\title{Simple Protocols for Oblivious Transfer and Secure
  Identification\\in the Noisy-Quantum-Storage Model}

\author{\vspace{-1.6cm}}

\author{Christian Schaffner$^1$\\
  \textit{$^1$ Centrum Wiskunde \& Informatica (CWI) Amsterdam, The Netherlands}} \date{\today}
\maketitle

\begin{abstract}
  We present simple protocols for oblivious transfer and
  password-based identification which are secure against general
  attacks in the noisy-quantum-storage model as defined
  in~\cite{KWW09arxiv}. We argue that a technical tool
  from~\cite{KWW09arxiv} suffices to prove security of the known
  protocols. Whereas the more involved protocol for oblivious transfer
  from~\cite{KWW09arxiv} requires less noise in storage to achieve
  security, our ``canonical'' protocols have the advantage of being
  simpler to implement and the security error is easier
  control. Therefore, our protocols yield higher OT-rates for many
  realistic noise parameters.

  Furthermore, the first proof of security of a direct protocol for
  password-based identification against general noisy-quantum-storage
  attacks is given.
\end{abstract}
\vspace{1mm}
\pagestyle{plain}

\section{Introduction}
Throughout history, a main goal of cryptography has been to provide
secure communication over insecure channels. In today's
internet-driven society however, more advanced tasks arise: people need
to do business and interact with peers they neither know nor trust. A
simple example is \emph{secure identification}: Users Alice and Bob
share a password $P$ and when setting up a communication, Alice wants
to make sure she is really interacting with Bob---the only other person
who knows $P$. Simply announcing $P$ is insecure, as any eavesdropper
can intercept $P$ and use it later to impersonate Bob. We need a
method to check whether two parties are in possession of the same
password, but without revealing any additional information.

Secure identification is a special case of the more general problem of
\emph{secure two-party computation}: Alice and Bob want to perform a
computation on private inputs in a way that they obtain the correct
result but no additional information about their inputs is
revealed. An interesting example are sealed-bit auctions where the
winner should be determined without opening the losing bids. Closer
to everyday life, almost any interaction with an Automated Teller
Machine (ATM) can be seen as an instance of secure two-party computation.

The techniques used in modern classical cryptography to secure
communication and provide secure two-party computation are based on
unproven mathematical assumptions such as the hardness of finding the
prime factors of large integer numbers (for example in the widely used
RSA scheme~\cite{RSA78}). We do not know any practical schemes which
are provably infeasible to break and it is unlikely that the currently
known mathematical techniques allow for such a scheme. In contrast,
quantum cryptography, which is based on transmitting information
stored in the state of single elementary particles, offers schemes
with \emph{provable security}.

\medskip The most prominent example is Quantum Key Distribution (QKD)
which allows two honest parties to securely communicate. In 1984,
Bennett and Brassard proposed a QKD protocol~\cite{BB84} which was
proven unconditionally secure~\cite{Mayers95, Yao95, SP00}. In other
words, security does not rely on any unproven assumptions but holds
against any eavesdropper Eve with unbounded (quantum) computing
power. Such provably secure key-distribution schemes cannot be
achieved by any classical means (without additional assumptions).
It is important to realize that the technical
requirements for honest parties to perform QKD protocols are well
within reach of current technology.  As of today, the technology has
even reached commercial level: At least three different companies are
selling hardware for QKD~\cite{SmartQuantum, idQuantique, MagicQ}.

After the discovery of QKD, researchers thought it was possible to use
quantum communication to implement more advanced cryptographic
primitives such as secure two-party computation. However, it was shown
in the late 90s that essentially \emph{no} cryptographic two-party
primitives can be realized if only a quantum channel is available and
no further restriction on the adversary is
assumed~\cite{Mayers97,LC97,Lo97}. In other words, secure two-party
computation is more difficult to achieve than key distribution. This
is not completely surprising given the generality of secure two-party
computation. Nevertheless, quantum cryptography might still help to
achieve significantly better schemes than purely classical
constructions.

\medskip Indeed, in joint work with Damg{\aa}rd, Fehr and Salvail, we
proposed in 2005 a new realistic assumption for quantum protocols
under which provably secure two-party computation becomes
possible~\cite{DFSS05}. The basic idea is to exploit the technical
difficulty of storing quantum information. In this
\emph{bounded-quantum-storage model}, security holds based on the sole
assumption that the parties' \emph{quantum memory during the execution
  of the protocol} is upper bounded. No further restrictions on the
(quantum) computing power nor the classical memory size are
assumed. Storing quantum information requires to keep the state of
very small physical systems such as single atoms or photons under
stable conditions over a long time. Building a reliable quantum memory
is a major research goal in experimental quantum
physics~\cite{JSCFP04,CMJLKK05,EAMFZL05,CDLK08,AFKLL08}. Despite these
efforts, current technology only allows storage times of at most a few
milliseconds.

Even though breaking the security of our protocols requires a large
quantum memory with long storage times, neither quantum memory nor the
ability to perform quantum computations are needed to actually run the
protocols; the technological requirements for honest parties are
comparable to QKD and hence well within reach of current
technology. Therefore, cryptographic schemes based on storage
imperfections provide potentially very useful solutions for secure
two-party computation with the advantage of much stronger security
guarantees compared to classical technology.

\subsection{Bounded- versus Noisy-Quantum-Storage Model}
In the bounded-quantum-storage model, we assume that a dishonest
receiver can perfectly store the incoming photons and perform perfect
quantum operations under the sole restriction that at a certain point
of the protocol, the size of his quantum memory is limited to a
constant fraction of the total number of received photons. Bounding
the size of the adversary's quantum storage in this way is a handy
assumption to work with in security proofs. In a series of works over
the last
years~\cite{DFSS05,DFRSS07,DFSS07,Schaffner07,DFSS08,DFSS10journal},
it has been shown that any type of secure two-party computation is
possible in the bounded-quantum-storage model.

On the other hand, simply limiting the adversary's quantum memory size
does not capture correctly the difficulty one currently faces when
trying to store photons. A better formalization of this difficulty is
to assume that the dishonest receiver uses the best available (but
still imperfect) photon-storage device. The imperfection of the
storage-device is modeled as noisy quantum channel where the noise
level of the channel increases with the amount of time during which
the quantum information needs to be stored. With current technology,
the noise reaches maximum level (i.e.~the quantum information is
completely lost) if a storage time in the order of milliseconds is
required~\cite{JSCFP04}.

First results in this \emph{noisy-quantum-storage model} have been
established in joint work with Terhal and
Wehner~\cite{WST08,STW09}. Assuming ``individual-storage
attacks''---where the adversary treats all incoming qubits in the same
way---the security of oblivious transfer and password-based
identification was established using the original protocols from the
bounded-quantum-storage model~\cite{DFRSS07,DFSS10journal}.

The most general storage attacks were first mentioned
in~\cite{Schaffner07}, but addressed only recently by K\"onig, Wehner
and Wullschleger \cite{KWW09arxiv}. In this most general model, the
adversary can for example try to use a quantum error-correcting code
in order to protect himself from storage errors. Concretely, he is
allowed to first perform an arbitrary perfect ``encoding attack'' on
the incoming quantum state, then he uses his (noisy) quantum-storage
device together with unlimited classical memory and finally, he can
again perform perfect quantum computations.\footnote{A detailed
  description of the model of~\cite{KWW09arxiv} will be given in
  Section~\ref{sec:noisystorage}, see also
  Figure~\ref{fig:noisystorage}.} The authors of~\cite{KWW09arxiv}
show how the security of protocols in this general model can be
related to the maximal rate of classical information that can be
transmitted over the noisy storage channel.

In more detail, \cite{KWW09arxiv} introduces the conceptual novelty of
splitting the security analysis of protocols for oblivious transfer
and bit commitment in two phases. In the first phase, the players use
the well-known BB84 quantum coding scheme to achieve a (quantum)
primitive which the authors call \emph{weak string erasure}. At the
end of this phase, the sender has a classical $n$-bit string $X$ and
the receiver holds an ``erased version'' of the string where a
uniformly random half of the bits of $X$ have been erased. Note that
this primitive is only classical for honest players, as a dishonest
receiver might hold quantum information about the sender's classical
output string.

For the second (purely classical) phase, they propose classical
reductions to build bit commitment and oblivious transfer based on
weak string erasure. Their approach to realize oblivious transfer is
quite involved. It uses interactive hashing~\cite{Savvides07}, for
which the standard classical protocol requires a lot of communication
rounds~\cite{NOVY98}\footnote{A constant-round variant of interactive
  hashing has been proposed in~\cite{DHRS04}. However, it is unclear
  how the weaker security guarantees affect the security proof
  in~\cite{KWW09arxiv}. The use of $\eta$-almost $t$-wise independent
  permutations might render this variant ``prohibitively complicated to implement
  in practice''~\cite{Savvides07}.}.  The analysis is
complicated by the fact that the dishonest receiver holds quantum
information, but can be handled by techniques of min-entropy sampling
developed by K\"onig and Renner~\cite{KR07}. It was left as open
question how to build password-based identification based on weak
string erasure or in general, secure against noisy-quantum-storage
attacks.

\subsection{Our Results and Outline of the Paper}
The main contribution of this paper is the insight that the new
technical tool derived in~\cite{KWW09arxiv} already suffices to prove
secure the original protocols from the bounded-quantum-storage model
for bit commitment, oblivious transfer~\cite{DFRSS07} and
password-based identification~\cite{DFSS07,DFSS10journal}. These original protocols
have the advantage that the classical post-processing is extremely
simpel. No communication-intensive protocols such as interactive
hashing are needed.

Comparing the protocol for oblivious transfer from~\cite{KWW09arxiv}
with our protocol, it turns out that the highly interactive
protocol~\cite{KWW09arxiv} can in theory be shown secure for less
noisy quantum-storage channels if infinitely many pulses are
available, i.e., security holds against a larger class of
adversarial receivers. However, the original protocols with the simpler analysis
presented here outperform the ones from~\cite{KWW09arxiv} in terms of
the security error. Thus, for a fixed number of pulses and a given
security threshold, the simpler protocols and our analysis yield
oblivious transfer of longer bit-strings most of the time.

We show for the first time the security against general noisy-storage
attacks of a direct protocol for password-based identification,
answering an open question posed in~\cite{KWW09arxiv}.

From a theoretical point of view, our insight shows that despite the
generality of the noisy-quantum-storage model, having the right tools
from~\cite{DFRSS07,KWW09arxiv} at hand, the protocols and security proofs do
not need to be much more complicated than in the conceptually simpler
bounded-quantum-storage model.

\begin{table}
\begin{tabular}[h]{|c|c|c|c|}
\hline
& Secure 1-2 OT~~ & ~~Canonical Protocol~~ &~~ Secure Identification ~~\\
\hline
\hline
\cite{WST08,STW09} & individual attacks & Yes & individual attacks\\
\hline
\cite{KWW09arxiv} & general attacks & No & No\\
\hline
This work & general attacks & Yes & general attacks\\
\hline
\end{tabular}
\caption{Summary of previous results in the noisy-quantum-storage
  model and the results presented here.
  \label{fig:comparison}}
\end{table}

\subsection{Outline of the Paper}
In Section~\ref{sec:prelim}, we define concepts and notation and
elaborate on the essential tool of min-entropy splitting in
Section~\ref{sec:ESL}. We present the noisy-quantum-storage and the
key ingredient from~\cite{KWW09arxiv} in
Section~\ref{sec:noisystorage}. Sections~\ref{sec:12OT},
\ref{sec:robust} and~\ref{sec:id} contain the security analyses for
oblivious transfer and password-based identification.

\section{Preliminaries} \label{sec:prelim}

We start by introducing the necessary definitions, tools and technical
lemmas that we need in the remainder of this text.

\subsection{Basic Concepts}
We use $\in_R$ to denote the uniform choice of an element from a set.
We further use $x|_{\setI}$ to denote the string $x=x_1,\ldots,x_n$
restricted to the bits indexed by the set $\setI \subseteq
\{1,\ldots,n\}$. For a binary random variable $C$, we denote by
$\ol{C}$ the bit different from $C$.


\paragraph{Classical-Quantum States}
A \emph{cq-state} $\rho_{XE}$ is a state that is partly classical,
partly quantum, and can be written as
$$
\rho_{XE}=\sum_{x \in \setX} P_X(x) \proj{x} \otimes \rho_E^x \, .
$$
Here, $X$ is a classical random variable distributed over the finite
set $\setX$ according to distribution $P_{X}$, $\set{\ket{x}}_{x \in
  \setX}$ is a set of orthonormal states and the register $E$ is in
state $\rho_E^x$ when $X$ takes on value $x$.


\paragraph{Conditional Independence.}
We also need to express that a random variable $X$ is (close to)
independent of a quantum state $\regE$ {\em when given a random
  variable $Y$}. This means that when given $Y$, the state $\regE$
gives no additional information on $X$. Formally, this is
expressed by requiring that $\rho_{X Y \regE}$ equals (or is close to)
$\rho_{X\leftrightarrow Y \leftrightarrow \regE}$, which is defined
as\footnote{The notation is inspired by the classical setting where
  the corresponding independence of $X$ and $Z$ given $Y$ can be
  expressed by saying that $X \leftrightarrow Y \leftrightarrow Z$
  forms a Markov chain. }
\begin{align} \label{eq:markovdefinition}
\rho_{X\leftrightarrow Y \leftrightarrow \regE} := \sum_{x,y}P_{X
  Y}(x,y)\proj{x} \otimes \proj{y} \otimes \rho_{\regE}^y \, .
\end{align}
In other words, $\rho_{X Y \regE} = \rho_{X\leftrightarrow Y
  \leftrightarrow \regE}$ precisely if $\rho_\regE^{x,y} =
\rho_\regE^{y}$ for all $x$ and $y$.  To further illustrate its
meaning, notice that if the $Y$-register is measured and value $y$ is
obtained, then the state $\rho_{X\leftrightarrow Y \leftrightarrow
  \regE}$ collapses to $ (\sum_{x}P_{X|Y}(x|y)\proj{x} )\otimes
\rho_{\regE}^y$, so that indeed no further information on $x$ can be
obtained from the $\regE$-register. This notation naturally extends to
$\rho_{X\leftrightarrow Y \leftrightarrow \regE|{\cal E}}$ simply by
considering $\rho_{X Y \regE|{\cal E}}$ instead of $\rho_{X Y
  \regE}$. Explicitly, $\rho_{X\leftrightarrow Y \leftrightarrow
  \regE|{\cal E}} = \sum_{x,y}P_{X Y|\ev}(x,y)\proj{x} \otimes
\proj{y} \otimes \rho_{\regE|\ev}^{y}$.

\paragraph{Non-uniformity}
We can say that a quantum adversary has little information about $X$
if the distribution $P_X$ given his quantum state is close to
uniform. Formally, this distance is quantified by the {\em
non-uniformity} of $X$ given $\rho_E = \sum_x P_X(x) \rho_E^x$
defined as
\begin{equation}
d(X|E) := \frac{1}{2}\left\|\,\mathbbm{1}/|\mX| \otimes \rho_E-\sum_{x}P_{X}(x) \proj{x} \otimes \rho_{E}^x\,\right\|_{1} \, .
\end{equation}
Intuitively,
$d(X|E) \leq \eps$ means that the distribution of $X$ is $\eps$-close to
uniform even given $\rho_E$, i.e., $\rho_E$ gives hardly any
information about $X$.
A simple property of the non-uniformity which follows from its
definition is that it does not change given independent information. Formally,
\begin{equation}
d(X|E,D) =d(X|E)
\label{eq:indep}
\end{equation}
for any cqq-state of the form $\rho_{XED}=\rho_{XE} \otimes \rho_D$.

 \subsection{Entropic Quantities}
Throughout this paper we use a number of entropic quantities. The
\emph{binary-entropy} function is defined as $h(p) \assign - p \log
p - (1-p) \log (1-p)$, where $\log$ denotes the logarithm to base 2
throughout this paper.

\subsubsection{(Conditional) Smooth Min-Entropy} 
We are concerned with the situation where an attacker holds quantum
information in register $E$ about a classical variable $X$,
described by a \emph{c}lassical-\emph{q}uantum state (cq-state) of the form
\[ \rho_{XE} = \sum_x P_X(x) \proj{x} \otimes \rho_E^x \, .
\]
We define the \emph{guessing probability of $X$ given $E$} as the
success probability of the best measurement carried out on $E$ in
order to guess $X$,
\begin{align*}
 \guess(X|E) \assign \max_{\set{M_x}} \sum_x P_X(x) \Tr(M_x \rho_E^x)
 \, ,
\end{align*}
where the maximisation is over all POVMs $\set{M_x}$ acting on
register $E$. The conditional min-entropy of $X$ given $E$ is defined
as $\hmin(X|E) \assign -\log \guess(X|E)$.

In case the adversary's information $E$ is described by a classical
variable $Y$, one can show that the guessing probibility becomes
$$
\guess(X|Y) \assign \sum_y P_Y(y) \max_x
P_{X|Y}(x|y) = \sum_y \max_x
P_{XY}(x,y) \, .
$$ 
More generally, we define $\hmin(X \ev|Y)$ for any event $\ev$ as $\hmin(X\ev|Y)
\assign - \log \big( \guess(X\ev|Y) \big)$ where%
\footnote{$\guess(X\ev|Y)$ can be understood as the optimal probability in guessing $X$ {\em and} have $\ev$ occur, when given~$Y$. }
$$
\guess(X\ev|Y) \assign \sum_y P_Y(y) \max_x P_{X\ev|Y}(x|y) = \sum_y \max_x P_{XY\ev}(x,y) \, .
$$ 
The \emph{conditional smooth min-entropy} $\hmine{\eps}{X|Y}$ is then
defined as
\[ \hmine{\eps}{X|Y} \assign \max_{\ev}
\hmin(X \ev|Y) 
\]
where the max is over all events $\ev$ with $P[\ev] \geq 1 - \eps$. 

Obviously, the unconditional versions of smooth and non-smooth
min-entropy are obtained by using a constant $Y$. Furthermore,
conditional smooth min-entropy can also be defined for quantum side
information, we refer to~\cite{Renner05,KWW09arxiv} for the formal
definitions.

In this paper, we will use the fact that smooth min-entropy obeys the
chain rule~\cite[Theorem 3.2.12]{Renner05}, i.e. for a ccq-state
$\rho_{XYE}$, we have
\begin{align} \label{eq:chain}
\hmin^\eps(X|YE) \geq \hmin^\eps(X|E) - \log |\mY| \, ,
\end{align}
where $|\mY|$ is the alphabet size of $Y$.

\subsection{Min-Entropy Splitting} \label{sec:ESL}
The key ingredients for the security proofs of both the 1-2 OT and
the secure identification schemes in~\cite{DFRSS07,DFSS07}
are uncertainty relations and variants of the \emph{min-entropy
splitting lemma}. In this section, we present an overview over the
variants known and derived for the bounded-quantum-storage model and
point out how they can be applied in the noisy-quantum-storage model.

If the joint entropy of two random variables $X_0$ and $X_1$ is large,
then one is tempted to conclude that at least one of $X_0$ and $X_1$
must still have large entropy, e.g.\ half of the original
entropy. Whereas such a reasoning is correct for Shannon entropy (it
follows easily from the chain rule and the fact that conditioning
does not increase the entropy), it
is in general incorrect for min-entropy.
There exist joint probability
distributions $P_{X_0 X_1}$ for which guessing $X_0$ and $X_1$
individually is easy, but guessing $X_0$ and $X_1$ simultaneously is
hard. Intuitively, for these distributions, guessing the value $x_i$
with the highest probability is easy, because the probabilities over
the other variable $X_{1-i}$ are uniform, but still sum up to a
significant mass.

However, the following basic version of the min-entropy splitting
lemma, which first appeared in a preliminary version
of~\cite{Wullschleger07} and was later developed further in the
context of randomness extraction~\cite{KR07}, shows that the intuition
about splitting the min-entropy {\em is} correct in a randomized
sense. This lemma (with a slightly different notion of min-entropy) is
used in the security proof of the 1-2 OT scheme in~\cite{DFRSS07}.
\begin{lemma}[Min-Entropy-Splitting Lemma~\cite{DFRSS07}]\label{lemma:ESLold}
  Let $\varepsilon \geq 0$, and let $X_0,X_1$ and $Z$ be random
  variables with \mbox{$\hmin^{\varepsilon}(X_0 X_1|Z) \geq \alpha$}.
  Then, there exists a random variable $D \in \set{0,1}$ such that
$$
\hmin^{\varepsilon}(X_D| D Z) \geq \alpha/2 -1 \, .
$$
\end{lemma}

\begin{proof}
 Let $\ev$ be an event such that $P[\ev] \geq 1 - \eps$ and 
 \begin{equation} \label{eq:12assumption} \sum_z P_Z(z) \cdot
   \max_{x_0,x_1} P_{X_0X_1\ev|Z}(x_0,x_1|z) \leq 2^{-\alpha} \, .
  \end{equation}
  By assumption, such an events exist.\footnote{In case $\eps = 0$,
    i.e., $\alpha$ lower bounds the ordinary (rather then the smooth)
    min-entropy, the $\ev$ is the events ``that always occurs'' and
    can be ignored from the rest of the analysis. } For a given $z$,
  we define $D$ to be 0 if and only if $P_{X_0|Z}(X_0|z) <
  2^{-\alpha/2}$. Then,
\begin{align} \begin{split} \label{eq:12ESL1}
  \sum_z P_Z(z) \cdot \max_{x_0} P_{X_0 D \ev |Z}(x_0,0|z) &\leq
  \sum_z P_Z(z) \cdot \max_{x_0} P_{X_0 D |Z}(x_0,0|z)\\
&=  \sum_z P_Z(z) \cdot \max_{x_0} P_{X_0|Z}(x_0|z) P_{D|X_0
    Z}(0|x_0,z) < 2^{-\alpha/2} \, ,
\end{split} \end{align}
because either $P_{X_0|Z}(x_0|z) < 2^{-\alpha/2}$ or $P_{D|X_0
  Z}(0|x_0,z) = 0$ by definition of $D$. On the other hand, we have
\begin{align} \begin{split} \label{eq:12ESL2}
  \sum_z P_Z(z) \cdot \max_{x_1} P_{X_1D\ev|Z}(x_1,1|z) &= \sum_z
  P_Z(z) \cdot \max_{x_1} \sum_{x_0} P_{X_0 X_1 D \ev|Z}(x_0,x_1,1|z)\\
  &\leq 2^{\alpha/2} \sum_z P_Z(z) \cdot \max_{x_0,x_1} P_{X_0
    X_1\ev|Z}(x_0,x_1|z) \leq 2^{-\alpha/2} \, ,
\end{split} \end{align} where the last inequality follows from the
assumption~\eqref{eq:12assumption} and the first is a consequence of
the fact that the number of non-zero summands (in the sum over $x_0$)
cannot be larger than $2^{\alpha/2}$, because for any $x_0$ with
$P_{X_0 X_1 D | Z}(x_0,x_1,1|z) > 0$, it also holds (by the definition
of $D$) that $P_{X_0|Z}(x_0|z) \geq 2^{-\alpha/2}$ and the sum over
all those $x_0$ would exceed 1 if there were more than $2^{\alpha/2}$
summands.

Combining \eqref{eq:12ESL1} and \eqref{eq:12ESL2}, we conclude 
that
\[ \guess(X_D \ev|D Z) = \sum_d \sum_z P_Z(z) \max_{x} P_{X_D D
  \ev|Z}(x_d,d|z) \leq 2 \cdot 2^{-\alpha/2} \, .
\]
The claim now follows by definition of $\hmin^\eps$.
\end{proof}

%
In order to prove the security of the identification scheme (see
Section~\ref{sec:id}), a more refined version of the min-entropy
splitting lemma was derived in \cite{DFSS10journal}. We reproduce it
here for convenience.
\begin{lemma}[Entropy-Splitting Lemma~\cite{DFSS10journal}]\label{lemma:basicES}
  Let $\eps \geq 0$. Let $X_1,\ldots,X_m$ and $Z$ be random variables
  such that $\hmin^{\eps}(X_i X_j|Z) \geq \alpha$ for all $i \neq
  j$. Then there exists a random variable $V$ over $\set{1,\ldots,m}$
  such that for any {\em independent} random variable $W$ over
  $\set{1,\ldots,m}$ with $\hmin(W) \geq 1$, 
$$
\hmin^{2 m\eps}(X_W|VWZ,V\!\neq\!W) \geq \alpha/2 - \log(m) 
 - 1 \, .
$$
\end{lemma}

\begin{proof}
  For any pair $i \neq j$ let $\ev_{ij}$ be an event such that
  $P[\ev_{ij}] \geq 1 - \eps$ and 
  \begin{equation} \label{eq:assumption} \sum_z P_Z(z) \cdot
    \max_{x_i,x_j} P_{X_iX_j\ev_{ij}|Z}(x_i,x_j|z) \leq 2^{-\alpha}
  \end{equation}
  for all $x_i \in \mX_i$, $x_j \in \mX_j$ and $z \in \cal Z$. By
  assumption, such events exist.\footnote{In case $\eps = 0$, i.e.,
    $\alpha$ lower bounds the ordinary (rather then the smooth)
    min-entropy, the $\ev_{ij}$ are the events ``that always occur''
    and can be ignored from the rest of the analysis. } For any
  $j=1,\ldots,m-1$ define
$$
L_j =
\Set{(x_1,\ldots,x_m,z)}{P_{X_1|Z}(x_1|z),\ldots,P_{X_{j-1}|Z}(x_{j-1}|z)
  < 2^{-\alpha/2} \wedge P_{X_j|Z}(x_j|z) \geq 2^{-\alpha/2}}
$$
Informally, $L_j$ consists of the tuples $(x_1,\ldots,x_m,z)$, where
$x_j$ has ``large'' probability given $z$ whereas all previous entries
have small probabilities. We define $V$ as follows. We let $V$ be the
index $j \in \set{1,\ldots,m-1}$ such that $(X_1,\ldots,X_m,Z) \in
L_j$, and in case there is no such $j$ we let $V$ be $m$. Note that if
there does exist such an $j$ then it is unique.

We need to show that this $V$ satisfies the claim. Fix $j \in
\set{1,\ldots,m}$. Clearly, for $i < j$,
\begin{align} \begin{split} \label{eq:ESL1}
  \sum_z P_Z(z) \cdot \max_{x_i} P_{X_i V \ev_{ij} |Z}(x_i,j|z) &\leq
  \sum_z P_Z(z) \cdot \max_{x_i} P_{X_i V |Z}(x_i,j|z)\\
&=  \sum_z P_Z(z) \cdot \max_{x_i} P_{X_i|Z}(x_i|z) P_{V|X_i
    Z}(j|x_i,z) < 2^{-\alpha/2} \, .
\end{split} \end{align}
Indeed, either $P_{X_i|Z}(x_i|z) < 2^{-\alpha/2}$ or $P_{V|X_i
  Z}(j|x_i,z) = 0$ by definition of $V$. 
Consider now $i > j$. Note that  
\begin{align} \begin{split} \label{eq:ESL2}
  \sum_z P_Z(z) \cdot \max_{x_i} P_{X_iV\ev_{ij}|Z}(x_i,j|z) &= \sum_z
  P_Z(z) \cdot \max_{x_i} \sum_{x_j} P_{X_i X_j
    V\ev_{ij}|Z}(x_i,x_j,j|z)\\
  &\leq 2^{\alpha/2} \sum_z P_Z(z) \cdot \max_{x_i,x_j} P_{X_i
    X_j\ev_{ij}|Z}(x_i,x_j|z) \leq 2^{-\alpha/2} \, ,
\end{split} \end{align}
where the last inequality follows from the
assumption~\eqref{eq:assumption} and the first is a consequence of the
fact that the number of non-zero summands (in the sum over $x_j$)
cannot be larger than $2^{\alpha/2}$, because for any $x_j$ with
$P_{X_i X_j V \ev_{ij} | Z}(x_i,x_j,j|z) > 0$, it also holds that
$P_{X_j|Z}(x_j|z) \geq 2^{-\alpha/2}$ and the sum over all those $x_j$
would exceed 1 if there were more than $2^{\alpha/2}$ summands.
Note that per-se, $\ev_{ij}$ is only defined in the probability space
given by $X_i$, $X_j$ and $Z$, but it can be naturally extended to the
probability space given by $X_1,\ldots,X_n,Z,V$ by assuming it to be
independent of anything else when given $X_i,X_j,Z$, so that e.g.\
$P_{X_i V \ev_{ij}|Z}$ is indeed well-defined.

Consider now an independent random variable $W$ with $\hmin(W) \geq 1$. 
By the assumptions on $W$ it holds that $P[V\!\neq\!W] \geq \frac12$
and $P_{X_W V W Z}(x_i,j,i,z) = P_{X_i V W Z}(x_i,j,i,z) = P_{X_i V
  Z}(x_i,j,z) P_W(i)$.  In the probability space determined by the
random variables $X_1,\ldots,X_n,V,W,Z$ and all of the events $\ev_{ij}$, define the event
$\ev$ as $\ev \assign \ev_{WV}$, so that $P_{X_W V W \ev|Z}(x_i,j,i|z) =
P_{X_i V W \ev_{ij}|Z}(x_i,j,i|z) = P_{X_i V \ev_{ij}|Z}(x_i,j|z)
P_W(i)$. Note that
$$
P[\bar{\ev}] = \sum_{i,j} P_{VW \bar{\ev}_{WV}}(j,i)
= \sum_{i,j} P_{V \bar{\ev}_{ij}}(j) P_W(i)
\leq \sum_{i,j} P[\bar{\ev}_{ij}] P_W(i)
\leq m \eps
$$
and thus $P[\bar{\ev}|V\!\neq\!W] \leq P[\bar{\ev}]/P[V\!\neq\!W]
\leq 2m \eps$.  From the above, it follows that
\begin{align*}
  \guess&(X_W,\ev|VWZ,V\neq W) = \sum_{z,i,j} \max_x P_{X_W V W Z \ev|V \neq
    W}(x,j,i,z)
  \leq 2 \sum_{z,i \neq j} \max_x P_{X_W V W Z \ev}(x,j,i,z)\\
  &=2 \sum_{z,i \neq j} P_Z(z) \cdot \max_x P_{X_W V W \ev|Z}(x,j,i|z)
  = 2 \sum_{z,i \neq j} P_Z(z) \cdot \max_{x_i} P_{X_i V \ev_{ij}|Z}(x_i,j|z) \cdot P_W(i) \\
  &= 2 \sum_{i} P_W(i) \sum_{j\neq i} \sum_z P_Z(z) \cdot
  \max_{x_i} P_{X_i V \ev_{ij}|Z}(x_i,j|z) \leq 2m \cdot 2^{-\alpha/2} \, ,
\end{align*}
where we used \eqref{eq:ESL1} and \eqref{eq:ESL2} in the last
inequality. The claim now follows by definition of $\hmin^\eps$.
\end{proof}

\subsection{Quantum Uncertainty Relation. } At the very core of our
security proofs lies (a special case of) the quantum uncertainty
relation from~\cite{DFRSS07}\footnote{In~\cite{DFRSS07}, a stricter
  notion of conditional smooth min-entropy was used, which in
  particular implies the bound as stated here. }, that lower bounds
the (smooth) min-entropy of the outcome when measuring an arbitrary
$n$-qubit state in a random basis $\theta \in \set{0,1}^n$.

\begin{theorem}[Uncertainty Relation~\cite{DFRSS07}]\label{thm:uncertainty}
Let $\regE$ be an arbitrary fixed $n$-qubit state. Let
$\Theta$ be uniformly distributed over $\set{+,\times}^n$ (independent
of $\regE$), and let $X\in\{0,1\}^n$ be the random variable for the outcome of
measuring $\regE$ in basis~$\Theta$.
Then, for any $\delta > 0$, the conditional smooth min-entropy is lower bounded by
$$
\hmin^\eps(X|\Theta) \geq \Big(\frac12 - 2\delta\Big) n 
$$
with $\eps \leq 2^{-\sigma(\delta)n}$ and 
\begin{equation} \label{eq:sigma}
\sigma(\delta) \assign
\frac{\delta^2 \log(e)}{32(2-\log(\delta))^2} \, .
\end{equation}
\end{theorem}

\subsection{Privacy Amplification}
We will make use of two-universal hash functions. A class $\setF$ of
functions $f: \01^n \rightarrow \01^\ell$ is called two-universal, if
for all $x\neq y \in \01^n$, we have $\Pr_{f \in_R \setF}[f(x) = f(y)]
\leq 2^{-\ell}$~\cite{CW79}.
The following theorem expresses how the application of hash functions
increases the privacy of a random variable X given a quantum
adversary holding $\rho_E$, the function $F$ and a classical random
variable $U$:

\begin{theorem}[\cite{Renner05,DFRSS07}] \label{thm:PA} Let
  $\setF$ be a class of two-universal hash functions from $\01^n$ to
  $\01^\ell$.  Let $F$ be a random variable that is uniformly and
  independently distributed over $\setF$, and let $\rho_{XUE}$ be a
  ccq-state. Then, for any $\eps \geq 0$,
$$
d(F(X)|F,U,E) \leq 2^{-\frac{1}{2}\left(\hminee{X|UE} - \ell
    \right)-1}+ \eps \, .
$$
\end{theorem}

\section{The Noisy-Quantum-Storage Model} \label{sec:noisystorage} 
The \emph{noisy-quantum-storage model} has been established
in~\cite{WST08,STW09} for the special case where the dishonest
receiver is limited to so-called ``individual-storage attacks'', i.e.~he
treats every incoming pulse independently (akin to individual attacks
in QKD).

The most general setting considered here is exactly the one described
in detail in~\cite[Sections~1.3 and~3.3]{KWW09arxiv}, see
Figure~\ref{fig:noisystorage} for an illustration.  The cheating
receiver is computationally unbounded, has unlimited classical storage
and can perform perfect quantum operations. If the protocol instructs
parties to wait for time $\Delta t$, a dishonest player has to discard
all quantum information, except for what he can encode arbitrarily
into his (noisy) quantum storage. This storing process is formally
described by a completely positive and trace-preserving (CPTP) map
$\cF:\cB(\cH_{in})\rightarrow\cB(\cH_{out})$. 

\begin{figure}[t]
\begin{center}
\includegraphics{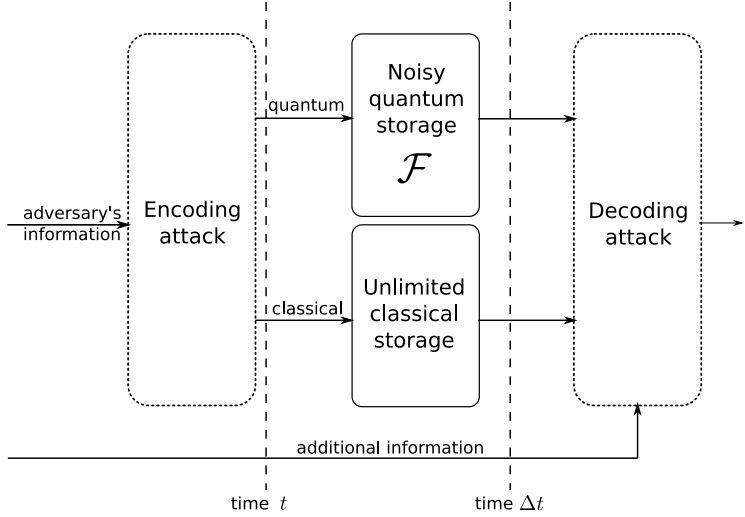}
\caption{(from~\cite{WCSL10}): During waiting times $\Delta t$, the
  adversary must use his noisy quantum storage described by the CPTP
  map $\cF$. Before using his quantum storage, he performs any
  (error-free) ``encoding attack'' of his choosing, which consists of
  a measurement or an encoding into an error-correcting code. After
  time $\Delta t$, he receives some additional information that he can
  use for decoding.} \label{fig:noisystorage}
\end{center}
\end{figure}

As in~\cite{KWW09arxiv}, let
\begin{align}
P_{succ}^{\cF}(n):=\max_{\{D_x\}_x,\{\rho_x\}_x}\frac{1}{2^n}\sum_{x\in
  \{0,1\}^n} \Tr(D_x\cF(\rho_x))\ \label{eq:succ}
\end{align}
be the maximal success probability of correctly decoding a randomly
chosen $n$-bit string $x\in\sbin^n$ sent over the quantum channel
$\cF$. Here, the maximum is over families of code states
$\{\rho_x\}_{x\in\sbin^n}$ on $\cH_{in}$ and decoding POVMs
$\{D_x\}_{x\in\sbin^n}$ on $\cH_{out}$.

Intuitively, if the quantum channel $\cF$ does not allow to transmit
enough classical information over it, we should be able to prove
security against a dishonest Bob with such a storage channel. Indeed,
the following two lemmas from~\cite{KWW09arxiv} formalize this
intuition and are the key ingredients to connect the security of
protocols in the \nsw\ for such channels with their ability to
transmit classical information.

\begin{lemma}[\cite{KWW09arxiv}]\label{lem:minentropygeneration}
Consider an arbitrary cq-state $\rho_{XQ}$ and
a CPTP map $\cF:\cB(\cH_Q)\rightarrow\cB(\cH_{out})$. Then,
 $\hmin(X|\cF(Q))\geq -\log P^{\cF}_{succ}(\lfloor \hmin(X)\rfloor)$.
\end{lemma}

\begin{lemma}[\cite{KWW09arxiv}]\label{lem:basicminentropyincrease}
  Consider an arbitrary ccq-state $\rho_{XTQ}$, and let $\varepsilon,
  \varepsilon'\geq 0$ be arbitrary. Let
  $\cF:\cB(\cH_Q)\rightarrow\cB(\cH_{Q_{out}})$ be an arbitrary CPTP map.
  Then,
\begin{align*}
\hmin^{\varepsilon+\varepsilon'}(X|T\cF(Q))\geq -\log P^{\cF}_{succ}\left(\big\lfloor \hmin^{\varepsilon}(X|T)-\log\frac{1}{\varepsilon'}\big\rfloor\right)\ .
\end{align*}
\end{lemma}

We are interested in channels~$\cN$ which satisfy the following
\emph{\scp}: The success probability~\eqref{eq:succ} decays
exponentially for rates $R$ above the capacity, i.e., it takes the
form
\begin{align}
P_{succ}^{\cN^{ \otimes n}}(nR)\leq 2^{-n\gamma^\cN(R)}\qquad\textrm{ where}\qquad \gamma^\cN(R)>0\textrm{ for all }R>C_\cN .\label{eq:strongconverseproperty} 
\end{align}
In~\cite{KW09}, property~\eqref{eq:strongconverseproperty} was shown
to hold for a large class of channels. An important example for which
we obtain security is the $d$-dimensional depolarizing
channel~$\cN_r:\cB(\mathbb{C}^d)\rightarrow\cB(\mathbb{C}^d)$ defined
for $d \geq 2$ as
\begin{align}
\cN_r(\rho) := r \rho + (1-r) \frac{\mathbbm{1}}{d}\ \qquad\textrm{ for some fixed } 
0 \leq r\leq 1\ , \label{eq:depolarizingchannel}
\end{align}
which replaces the input state $\rho$ with the completely mixed state
with probability~$1-r$. For $d=2$, having storage channel
$\cN_r^{\otimes n}$ means that the adversary can store $n$ qubits
which are affected by independent and identically distributed
noise. To see for which values of $r$ we can obtain security, we need
to consider the classical capacity of the depolarizing channel as
evaluated by King~\cite{King03}. For $d=2$, i.e., qubits, it is given
by
\begin{align*}
C_{\cN_{r}}=1+\frac{1+r}{2}\log\frac{1+r}{2}+\frac{1-r}{2}\log\frac{1-r}{2}\ . 
\end{align*}

\section{1-2 Oblivious Transfer} \label{sec:12OT}

\subsection{Security Definition and Protocol}

In this section we prove the security of a randomized version of 1-2
OT (Theorem \ref{thm:secure1}) from which we can easily obtain 1-2
OT. In such a randomized 1-2 OT protocol, Alice does not input two
strings herself, but instead receives two strings $S_0$, $S_1 \in
\01^\ell$ chosen uniformly at random. Randomized OT (ROT) can easily
be converted into OT. After the ROT protocol is completed, Alice uses
her strings $S_0,S_1$ obtained from ROT as one-time pads to encrypt
her original inputs $\hat{S_0}$ and $\hat{S_1}$, i.e.~she sends an
additional classical message consisting of $\hat{S_0} \oplus S_0$ and
$\hat{S_1} \oplus S_1$ to Bob. Bob can retrieve the message of his
choice by computing $S_C \oplus (\hat{S}_C \oplus S_C) =
\hat{S}_C$. He stays completely ignorant about the other message
$\hat{S}_{\ol{C}}$ since he is ignorant about $S_{\ol{C}}$. The
security of a quantum protocol implementing ROT is formally defined in
\cite{DFRSS07} and justified in~\cite{FS09} (see
also~\cite{WW08}).

\begin{definition} \label{def:ROT}
An $\eps$-secure 1-2 $\mbox{ROT}^\ell$ is a protocol between Alice
and Bob, where Bob has input $C \in \01$, and Alice has no input.
\begin{itemize}
\item (Correctness) If both parties are honest, then for any
  distribution of Bob's input $C$, Alice gets outputs $S_0,S_1 \in
  \01^\ell$ which are $\eps$-close to uniform and independent of $C$
  and Bob learns $Y = S_C$ except with probability $\eps$.
\item (Security against dishonest Alice) If Bob is honest and obtains output $Y$,
  then for any cheating strategy of Alice resulting in her state
  $\rho_A$, there exist random variables $S'_0$ and $S'_1$ such that
  $\Pr[Y=S'_C] \geq 1- \eps$ and $C$ is independent of $S'_0$,$S'_1$
  and $\rho_A$\footnote{Existence of the random variables
    $S'_0,S'_1$ has to be understood as follows: given the cq-state
    $\rho_{Y \! A}$ of honest Bob and dishonest Alice, there exists a
    cccq-state $\rho_{Y S'_0 S'_1 A}$ such that tracing out the
    registers of $S'_0,S'_1$ yields the original state $\rho_{Y A}$
    and the stated properties hold.}.
\item (Security against dishonest Bob) If Alice is honest, then for any cheating strategy of Bob resulting in his state $\rho_B$,
there exists a random variable $D \in \01$ such that $d(S_{\ol{D}}|S_{D}D\rho_B) \leq \eps$.
\end{itemize}
\end{definition}

We consider the same protocol for ROT as in~\cite{BBCS91,DFLSS09}. 

\begin{protocol}[\cite{BBCS91,DFLSS09}]1-2 $\mbox{ROT}^\ell$ \label{prot:nonoise}
\begin{enumerate}
\item Alice picks $x \in_R \01^n$ and $\theta \in_R \{+,\times\}^n$.
  At time $t=0$, she sends
  $\ket{x_1}_{\theta_1},\ldots,\ket{x_n}_{\theta_n}$ to Bob.
\item Bob picks $\hat{\theta} \in_R \{+,\times\}^n$ at random and
measures the $i$th qubit in the basis $\hat{\theta}_i$. 
He obtains outcome $\hat{x} \in \01^n$.
\item[]\hspace{-1cm}Both parties wait time $\Delta t$.
\item Alice sends the basis information $\theta=\theta_1,\ldots,\theta_n$ to Bob.
\item Bob, holding choice bit $c$, forms the sets $\setI_c = \{i\in
  [n] \mid \theta_i = \hat{\theta}_i\}$ and $\setI_{1-c} = \{i \in [n]
  \mid \theta_i \neq \hat{\theta}_i\}$. He sends $\setI_0,\setI_1$ to
  Alice.
\item Alice picks two hash functions $f_0,f_1 \in_R \setF$, where
  $\setF$ is a class of two-universal hash functions. She sends
  $f_0$,$f_1$ to Bob.  Alice outputs
  $s_0 = f_0(x|_{\setI_0})$ and $s_1 =
  f_1(x|_{\setI_1})$ \footnote{If $x|_{\setI_b}$ is less
    than $n$ bits long Alice pads the string $x|_{\setI_b}$ with 0's
    to get an $n$ bit-string in order to apply the hash function to
    $n$ bits.}.
\item Bob outputs $s_c = f_c(\hat{x}|_{\setI_c})$.
\end{enumerate}
\end{protocol}

\subsection{Security Analysis}

\paragraph{Correctness}
First of all, note that it is clear that the protocol fulfills its
task correctly.  Bob can determine the string $x|_{\setI_c}$ (except
with negligible probability $2^{-n}$ the set ${\cal I}_c$ is
non-empty) and hence obtains $s_c$.  Alice's outputs $s_0,s_1$
are perfectly independent of each other and of $c$.

\paragraph{Security against Dishonest Alice}
Security holds in the same way as shown in~\cite{DFRSS07}. Alice
cannot learn anything about Bob's choice bit from the index information
$\setI_0,\setI_1$ she receives, and Alice's input strings can be
extracted by letting her interact with an unbounded receiver.

\paragraph{Security against Dishonest Bob}

Proving that the protocol is secure against Bob requires more work.
Our goal is to show that there exists a $D \in \{0,1\}$ such that Bob
with noisy storage as described in Section~\ref{sec:noisystorage} is
completely ignorant about $S_{\ol{D}}$. Since we are
performing $1$-out-of-$2$ oblivious transfer of $\ell$-bit strings,
$\ell$ corresponds to the ``amount'' of oblivious transfer we can
perform for a given security parameter $\eps$ and number of qubits
$n$.

\begin{theorem} \label{thm:secure1} Fix $0<\delta<\frac14$ and let
\begin{align} \label{eq:epsilon}
  \eps = 2 \exp\left( -
    \frac{(\delta/4)^2}{32(2+\log\frac{4}{\delta})^2} \cdot n \right)
  \, .
\end{align}
Then, for any attack of a dishonest Bob with storage
$\cF:\cB(\cH_{in})\rightarrow\cB(\cH_{out})$,
Protocol~\ref{prot:nonoise} is $2\eps$-secure against a dishonest
receiver Bob according to Definition~\ref{def:ROT}, if $n \geq
4/\delta$ and
$$
\ell \leq -\frac12 \log P^{\cF}_{succ}\left( \left(\frac{1}{4} -
    \delta \right)n \right)- \log\left(\frac{1}{\eps}\right) \, .
$$ 
\end{theorem}

\begin{proof}
  We need to show the existence of a binary random variable $D$ such
  that $S_{\bar{D}}$ is $\eps$-close to uniform from Bob's point of
  view. 

We can argue as in the proof of the security of weak string erasure
for honest Alice (Section~3.3 in~\cite{KWW09arxiv}) that 
\[ \hmine{\eps/2}{X_0 X_1 | \Theta K} \geq \frac{n}{2} - \frac{n
  \delta}{2} \, ,
\]
where $K$ denotes Bob's classical information obtained from the encoding attack.
Classical min-entropy splitting (Lemma~\ref{lemma:ESLold}) then
ensures that there exists a binary random variable $D \in \set{0,1}$
such that
\[ \hmine{\eps/2}{X_{\ol{D}} |D \Theta K } \geq \frac{n}{4} - \frac{n
  \delta}{4} -1 \, .
\]

One can now continue to argue as in the proof of Theorem 3.3
in~\cite{KWW09arxiv}, i.e.~we use Lemma~\ref{lem:basicminentropyincrease}
to get
\begin{align*}
  \hmine{\eps}{X_{\ol{D}} | D \Theta K Q_{out}} \geq -\log
  P^{\cF}_{succ}\left( \frac{n}{4} - \frac{n \delta}{4} -1 -\log
    \frac{2}{\eps} \right)
  \geq  -\log
  P^{\cF}_{succ}\left( \left(\frac{1}{4} - \delta \right)n \right) \, ,
\end{align*}
where the last step follows in the same way as in~\cite{KWW09arxiv} from
the monotinicity of the success probability $P^{\cF}_{succ}(m) \leq
P^{\cF}_{succ}(m')$ for $m \geq m'$ and the fact that $\log \frac{2}{\eps}
\leq \frac{\delta}{2} n \leq \frac{3\delta}{4}n
-1$. 

The rest of the security proof is analogous to the proof
in~\cite{DFRSS07}: It follows from the chain rule
for smooth min-entropy~\eqref{eq:chain} that
\begin{align*}
\hmine{\eps}{X_{\ol{D}}|D \Theta S_D K Q_{out}}
&\geq \hmine{\eps}{X_{\ol{D}} S_D|D \Theta K Q_{out}} - \ell \\
&\geq -\log
  P^{\cF}_{succ}\left( \left(\frac{1}{4} - \delta \right)n \right) - \ell.
\end{align*}
The privacy amplification Theorem~\ref{thm:PA} yields
\begin{equation}\label{eq:explicitTradeoff}
  d(F_{\ol{D}}(X_{\ol{D}}) \mid D \Theta F_D S_D K Q_{out}) \leq
  2^{-\frac12( -\log P^{\cF}_{succ}\left( \left(\frac{1}{4} - \delta \right)n \right) - 2\ell) }+ \eps
  \,
\end{equation}
which is smaller than $2 \eps$ as long as
$$
-\frac12 \log P^{\cF}_{succ}\left( \left(\frac{1}{4} - \delta \right)n \right)-
\ell \geq \log\left(\frac{1}{\eps}\right) \, .
$$
from which our claim follows.
\end{proof}

\subsection{Tensor-product channels}

\begin{corollary} \label{cor:tensorproductOT}
Let Bob's storage be described by $\cF = \cN^{\otimes \nu n}$ with
$\nu >0$, where $\cN$ satisfies the strong-converse
property~\eqref{eq:strongconverseproperty}, and 
$$ C_{\cN} \cdot \nu < \frac14 \, .$$
Fix $\delta \in ]0,\frac14 - C_{\cN} \cdot \nu[$, and let $\eps$ be
defined as in~\eqref{eq:epsilon}. Then, for any attack of a dishonest Bob,
Protocol~\ref{prot:nonoise} is $2\eps$-secure against a dishonest
receiver Bob according to Definition~\ref{def:ROT}, if $n \geq
4/\delta$ and
$$
\ell \leq  \gamma^{\cN}\left(\frac{1/4-\delta}{\nu}\right) \cdot
\frac{\nu n }{2} - \log \left(\frac{1}{\eps} \right) \, .
$$
\end{corollary}
\begin{proof}
We can substitute $n$ by $\nu n$ and $R$ by $R/\nu$ in the
strong-converse property~\eqref{eq:strongconverseproperty} to obtain
\begin{align*}
-\frac{1}{n}\log P^{\cN^{\otimes \nu n}}_{succ}(nR)\geq \nu \cdot \gamma^{\cN}(R/\nu)\ .
\end{align*}
The claim then follows from Theorem~\ref{thm:secure1} by setting~$R:=\frac{1}{4}-\delta$.
\end{proof}

For the $d$-dimensional depolarizing channel
\begin{equation} \label{eq:depolChannel}
\mN_r(\rho) = r \rho + (1-r) \frac{\mathbbm{1}}{d}.
\end{equation}
which preserves a $d$-dimensional input state with probability $r$ and
depolarizes it completely with probability $1-r$, it has been shown
in~\cite{KW09,KWW09arxiv} that
$$
\gamma^{\cN}(R) = \max_{\alpha \geq 1} \frac{\alpha-1}{\alpha}
\left(R - \log d + \frac{1}{1-\alpha} \log \left(\left(r + \frac{1-r}{d}\right)^\alpha 
+ (d-1)\left(\frac{1-r}{d}\right)^\alpha \right) \right)\ .
$$

We compare the parameters in terms of OT- and error-rate of our
approach to the ones in~\cite{KWW09arxiv}. In
Figure~\ref{fig:nurpossible}, the regions of the noise-parameter $r$
and storage-rate $\nu$ from our approach (red) and the
\cite{KWW09arxiv}-approach (blue) are shown. As the information
rate after min-entropy splitting in our approach is lower than without
min-entropy splitting, the range of noisy storage channels for which
security can theoretically be shown is smaller in our
approach. However, we will see in the following that the error
overhead due to the complicated post-processing with interactive
hashing in~\cite{KWW09arxiv} nullifies that advantage again.

\begin{figure}[t]
\begin{center}
\scalebox{1}{
\begin{pspicture}(0,0)(8.0,8.0)
\psset{unit=.7cm}
\psset{linewidth=.8pt}
\psset{labelsep=2.5pt}
\put(0,0){\epsfig{file=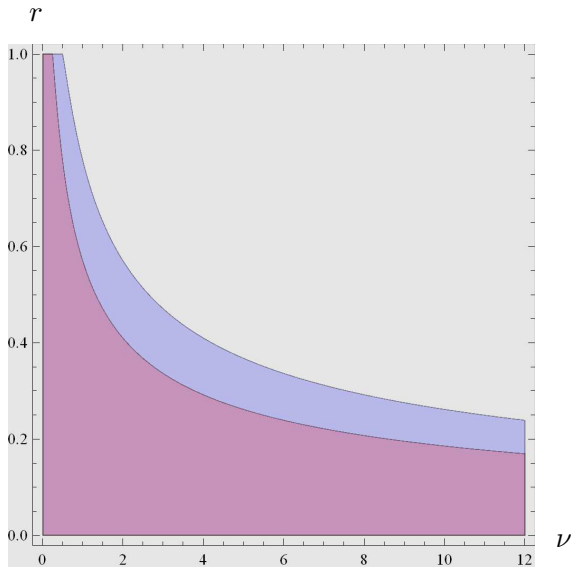,width=7cm}}
\put(0.4,10.4){$r$}
\put(10.4,0.4){$\nu$}
\end{pspicture}
}
\caption{Possible regions of a depolarizing qubit channel with noise
  parameter $r$ and storage rate $\nu$ where security for OT can be
  established for asymptotically many pulses. The
  \cite{KWW09arxiv}-approach yields the blue region, whereas our simpler approach gives the red subset of it.
  \label{fig:nurpossible}}
\end{center}
\end{figure}

We investigate two scenarios, in both of which we are ready to accept
a security error of at most $10^{-8}$.  In the first scenario, we are
given $n=10^{10}$ pulses to work with against an adversary with
depolarizing qubit channel ($d=2$) with noise rate $r$ and storage
rate $\nu=1$. In our approach, according to
Corollary~\ref{cor:tensorproductOT}, the security error is $2\eps$
where $\eps$ is defined in~\eqref{eq:epsilon}, thus for $n=10^{10}$,
we can choose $\delta=0.0106$ to have the error small enough. The
resulting OT-rate $\ell/n$ is the red line in Figure~\ref{fig:1010}
for different noise rates $r$ and a storage rate of $\nu=1$.  In the
approach of~\cite{KWW09arxiv}, the security error is harder to control
as it also depends on other parameters such as the noise rate $r$ and
a new parameter $\omega$. In order to keep it below the required
$10^{-8}$, we choose $\delta=0.011$ and $\omega=2$. The resulting
OT-rate is plotted as blue dashed line in Figure~\ref{fig:1010}. Note
that this amount of pulses are not sufficient to keep the security
error below $10^{-8}$ for noise rates $r$ above $0.21$.

In Figure~\ref{fig:1015}, we investigate the same setting but with
many more pulses, namely $n=10^{15}$. With that many pulses, the error
is better to control in the~\cite{KWW09arxiv}-approach and leads to
higher OT-rates compared to our approach for noise parameters between
$0.34 < r < 0.52$. In all other cases, our simpler approach allows to
get OT of longer strings while keeping the security error below
$10^{-8}$. 

To put these numbers of pulses into perspective, one can think of a
weak-coherent pulse setup which runs at 1GHz and emits a single
photons with Poisson distribution with parameter $\mu=1$, i.e.~with
probability $e^{-\mu}\mu  \approx 0.3679$ per pulse. Hence, we have to wait
approximately 27~seconds to obtain $n=10^{10}$ single
pulses, whereas it takes $10^6 \cdot e$ seconds, i.e.~roughly 30 days to generate
$n=10^{15}$ single pulses.

\begin{figure}[t]
\begin{minipage}[t]{0.45\textwidth}
\begin{center}
\scalebox{0.7}{
\begin{pspicture}(0,0)(12.0,5.0)
\psset{unit=.7cm}
\psset{linewidth=.8pt}
\psset{labelsep=2.5pt}
\put(0,0){\epsfig{file=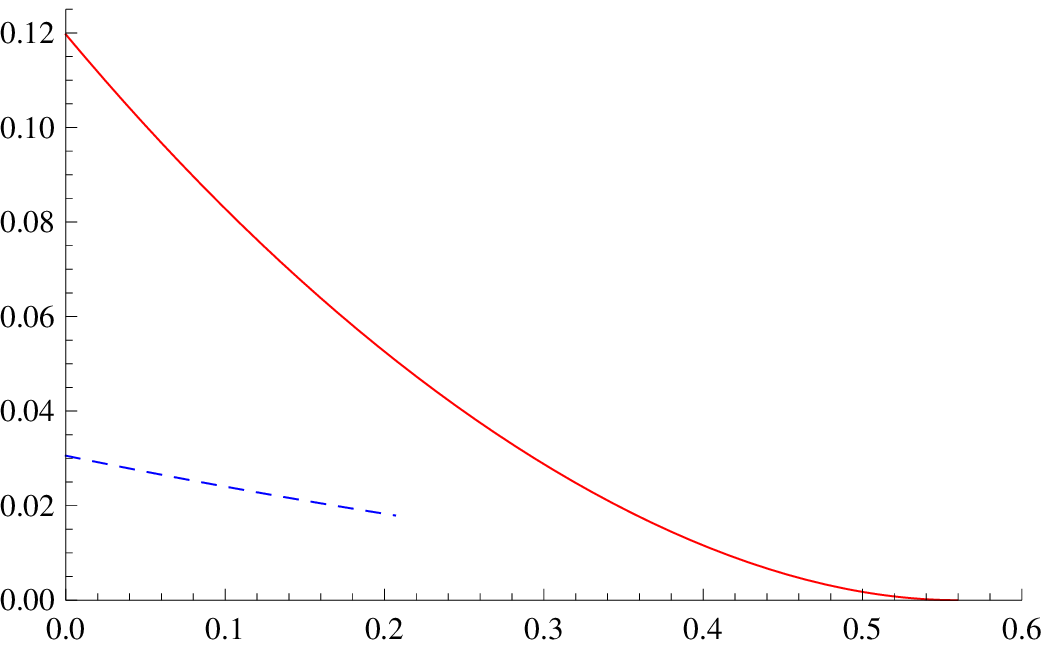,width=10cm}}
\put(0.4,9.2){$\ell/n$}
\put(14.4,0.4){$r$}
\end{pspicture}
}
\caption{The adversary's storage is depolarizing qubit noise $\cF=\cN_r^{\otimes
    n}$ with $d=2$, $\nu=1$, and $n=10^{10}$. 
The horizontal axis represents the noise parameter $r$, 
while the vertical axis represents the OT-rate $\ell/n$. The rates are
only plotted for regions where the security error stays below $10^{-8}$.
The red line represents the OT-rate obtained from our approach 
(Corollary~\ref{cor:tensorproductOT} with $\delta=0.0106$). The dashed
blue line is the rate from the \cite{KWW09arxiv}-approach with
optimised extra parameters $\delta=0.011$ and $\omega=2$. For
$r>0.21$, the security error is above the allowed threshold
$10^{-8}$. For this many pulses, our approach provides a higher
OT-rate for all possible noise parameters $r$ while keeping the security error reasonably low.
\label{fig:1010}}
\end{center}
\end{minipage}\hspace{5mm}
\begin{minipage}[t]{0.45\textwidth}
\begin{center}
\scalebox{0.7}{
\begin{pspicture}(0,0)(12.0,5.0)
\psset{unit=.7cm}
\psset{linewidth=.8pt}
\psset{labelsep=2.5pt}
\put(0,0){\epsfig{file=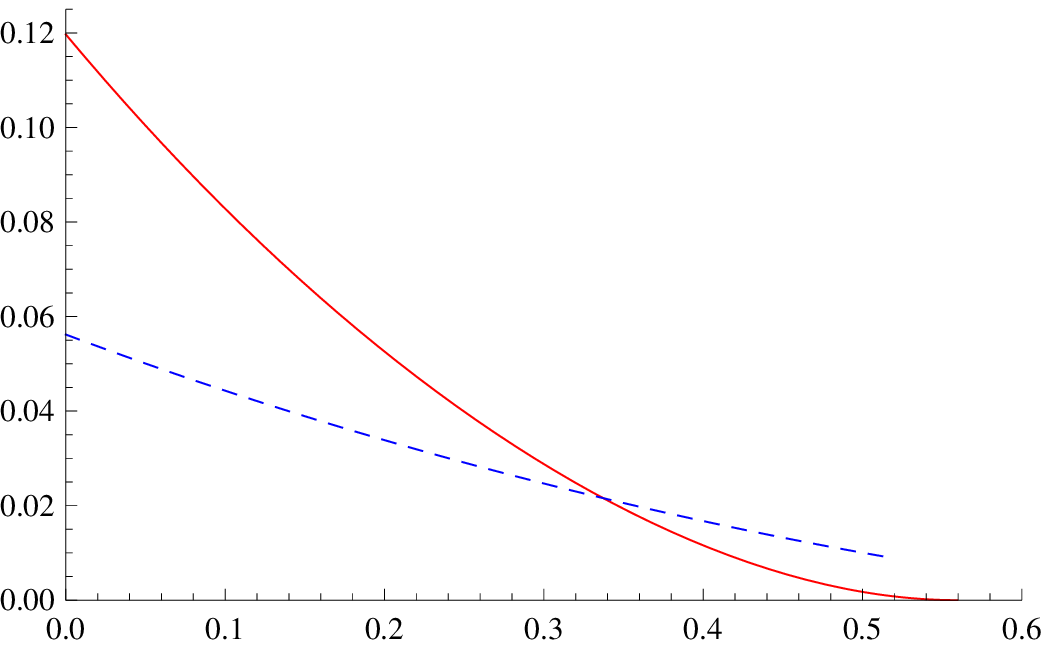,width=10cm}}
\put(0.4,9.2){$\ell/n$}
\put(14.4,0.4){$r$}
\end{pspicture}
}
\caption{As in Figure~\ref{fig:1010}, but for many more pulses, namely
  $n=10^{15}$.  The red line represents the OT-rate obtained from our
  approach (Corollary~\ref{cor:tensorproductOT} with
  $\delta=0.000057588$). The dashed blue line is the rate from the
  \cite{KWW09arxiv}-approach with optimised extra parameters
  $\delta=0.0005$ and $\omega=10$. For $r>0.47$, the security error is
  above the allowed threshold $10^{-8}$. For noise parameters between
$0.34 < r < 0.52$, the \cite{KWW09arxiv}-approach yields higher
OT-rates. For all other noise rates $r$, our simpler approach yields
higher rates.
\label{fig:1015}}
\end{center}
\end{minipage}
\end{figure}

\section{Robust Oblivious Transfer}
\label{sec:robust} 
In a practical setting, imperfections in Alice's and Bob's apparatus
as well as in the communication channel manifest themselves in form of
erasures and bit-flip errors. This setting has been analyzed for individual
attacks in~\cite{STW09} and for general attacks
in~\cite{WCSL10}. 
In the following, we present an upgraded protocol for oblivious
transfer along the lines of~\cite{WCSL10} but with a much
simpler and natural post-processing. 

\subsection{Protocol}
We consider the same setup as in~\cite{WCSL10}. Before engaging
in the actual protocol, Alice and Bob agree on a security-error
probability $\eps>0$. The parameter $p^h_\lostB$ denotes the
probability that an honest Bob observes no click in his detection
apparatus and the corresponding parameter $\zeta^h_\lostB$ says how
much fluctuations we allow. Typically, we use a $\zeta^h_\lostB$ of
order $\sqrt{\ln(2/\eps)/(2n)}$ such that the Chernoff bound allows us
to argue that $p^h_\lostB$ lies in the interval
$[(p^h_{\lostB}-\zeta^h_\lostB)n,(p^h_{\lostB}+\zeta^h_\lostB)n]$
except with probability $\eps$.

Error-correction is done using a one-way (forward) error correction
scheme, e.g.~by using low-density parity-check (LDPC) codes. The
players agree on a linear code which can correct errors in a $k$-bit
string by announcing the syndrome of the string. If each bit of the
string is flipped independently with probability $p^h_{\errB}$, this procedure
amounts to sending error-correcting information of at most $1.2 \cdot
h(p^h_{\errB}) \cdot k$ bits~\cite{ELAB09}.

We assume that the players have synchronized clocks. In each time
slot, Alice sends one qubit to Bob.

\begin{protocol}Robust 1-2 $\mbox{ROT}^\ell(C,T,\eps)$ \label{prot:practical}
\begin{enumerate}
\item Alice picks $x \in_R \01^n$ and $\theta \in_R \{+,\times\}^n$
  uniformly at random.
\item Bob picks $\hat{\theta} \in_R \{+,\times\}^n$ uniformly at
  random.
\item For $i=1,\ldots,n$: In time slot $t=i$, Alice sends bit $x_i$
  encoded in basis $\theta_i$ to Bob.  

In each time slot, Bob measures
  the incoming qubit in basis $\hat{\theta}_i$ and records
  whether he detects a photon or not. He obtains some bit-string $\hat{x}
  \in \01^m$ with $m \leq n$.
\item Bob reports back to Alice in which time slots he recorded a click.
\item \label{step:alicecheck} Alice restricts herself to the set of $m
  < n$ bits that Bob did not report as missing. Let this set of qubits
  be $S_{\rm remain}$ with $|S_{\rm remain}|=m$. If $m$ does not lie
  in the interval
  $[(1-p^h_{\lostB}-\zeta^h_\lostB)n,(1-p^h_{\lostB}+\zeta^h_\lostB)n]$,
  then Alice aborts the protocol. 
\item[]\hspace{-1cm}Both parties wait time $\Delta t$.
\item Alice sends the basis information
  $\theta=\theta_1,\ldots,\theta_m$ of the remaining positions to Bob.
\item \label{step:firstuseofchoicebit} Bob, holding choice bit $c$, forms the sets $\setI_c = \{i\in [m] \mid \theta_i = \hat{\theta}_i\}$ and $\setI_{1-c} = \{i \in [m]
  \mid \theta_i \neq \hat{\theta}_i\}$. He sends $\setI_0, \setI_1$ to
  Alice.
\item Alice picks two two-universal hash functions $f_0,f_1 \in_R
  \setF$ and sends $f_0$,$f_1$ and the syndromes $\syn(x|_{\setI_0})$
  and $\syn(x|_{\setI_1})$ to Bob.  Alice outputs $s_0 =
  f_0(x|_{\setI_0})$ and $s_{1} = f_1(x|_{\setI_1})$.
\item Bob uses $\syn(x|_{\setI_c})$ to correct the errors on his
  output $\hat{x}|_{\setI_c}$. He obtains the corrected bit-string $x_{\rm
    cor}$ and outputs $s'_c = f_c(x_{\rm cor})$.
\end{enumerate}
\end{protocol}

\subsection{Security Analysis}

\paragraph{Correctness} If both players are honest, Bob reports back
enough rounds to Alice. Therefore, in Step~\ref{step:alicecheck} the
protocol is aborted with probability at most $\eps$. The
error-correcting codes are chosen such that Bob can decode except with
probability $\eps$.  These facts imply that if both parties are
honest, the protocol is correct except with probability $2 \eps$. 

\paragraph{Security against Dishonest Alice}
Even though in this scenario Bob {\em does} communicate to Alice, the
information about which qubits were erased is independent of Bob's
choice bit $c$ as this bit is only used in
Step~\ref{step:firstuseofchoicebit}. Hence Alice does not learn
anything about his choice bit $c$. Her input strings can be extracted
as in the analysis of Protocol~\ref{prot:practical}.

\paragraph{Security against Dishonest Bob}
In the previous Section~\ref{sec:12OT}, we have seen that the security
analysis for weak string erasure from~\cite{KWW09arxiv} essentially
carries over to 1-2 oblivious transfer. Similarly, the security
analysis for weak string erasure with errors from~\cite{WCSL10}
can be adapted to analyse Protocol~\ref{prot:nonoise}.

We will use the following probabilities: (see~\cite{WCSL10} for
details and some example parameters for concrete setups)
\begin{center}
\begin{tabular}{|l|l|}
\hline
$p^d_\lostB$ & dishonest Bob observes no click in his detection
apparatus\\
& (due to imperfections in Alice's apparatus)\\[2pt]
\hline
$p^h_\lostB$ & honest Bob observes no click in his detection
apparatus\\
& (due to losses and imperfections of both player's apparatus)\\[2pt]
\hline
$p^1_\sent$ & Alice sends exactly $1$ photon.\\[2pt]
\hline
$p^h_{\errB}$ & honest Bob outputs the wrong bit\\
&(due to misalignments and noise on the channel)\\[2pt]
\hline
\end{tabular}
\end{center}

\begin{theorem}[Security against dishonest Bob]\label{thm:robustOT}
Fix $0<\delta<\frac14$ and let
\begin{align} \label{eq:epsilon}
  \eps = 2 \exp\left( -
    \frac{(\delta/4)^2}{32(2+\log\frac{4}{\delta})^2} \cdot m^1 \right)
  \, .
\end{align}
Then, for any attack of a dishonest Bob with storage
$\cF:\cB(\cH_{in})\rightarrow\cB(\cH_{out})$,
Protocol~\ref{prot:practical} is $2\eps$-secure against a dishonest
receiver Bob according to Definition~\ref{def:ROT}, if $m^1 \geq
4/\delta$ and the length of the OT-strings
$$
\ell \leq -\frac12 \log P^{\cF}_{succ}\left( \left(\frac{1}{4} -
    \delta \right)m^1 \right)- 1.2 \cdot h(p^h_{\errB}) \cdot \frac{m}{2} - 
\log\left(\frac{1}{\eps}\right) \, ,
$$ 
where $m^1 \assign (p^1_\sent - p^h_\lostB + p^d_\lostB) n$ is the minimal number of
single-photon rounds remaining and $m = (1-p^h_{\lostB}) n$ is the total number of rounds remaining.

\end{theorem}

\begin{proof}
As in~\cite{WCSL10}, we adopt the conservative viewpoint that a dishonest Bob does not
experience any bit-errors nor losses on the channel. Furthermore, we assume that
a dishonest receiver can detect when multiple photons arrive and
extract the encoded bit without knowledge of the encoding basis. These
multi-photon rounds will thus not contribute to the uncertainty of a
dishonest Bob. He will also not keep any quantum information about
these bits.

The main complication in this more practical scenario is that a
dishonest Bob might falsely report back rounds as missing in order to
decrease the overall fraction of single-photon rounds where he has
uncertainty about the encoded bits. 

Let $p^h_{\lostB}$ be the probability that honest Bob does not
register a click (due to losses in the channel and imperfect apparatus
of both players). On the other hand, let $p^d_{\lostB}$ be the
probability that a dishonest Bob does not register a click (due to
imperfections in Alice's apparatus). We assume that a dishonest Bob
will always report a round as missing if he did not register a click
(because there is no advantage for him not doing so). We also assumed
that Bob gets full information when more than one photon was sent and
hence, he will not report these rounds as missing. We conclude that
out of the $n$ rounds, dishonest Bob will report the maximal amount of
$(p^h_{\lostB} - p^d_{\lostB})n$ \emph{single-photon} rounds as
missing. That means that of the total $m=(1-p^h_{\lostB}) n$ rounds
that Alice accepts, at least 
\begin{align}
  m^1 \assign (p^1_{\sent} - (p^h_{\lostB} - p^d_{\lostB}))n
\end{align} are single-photon rounds. 

It can be argued as in~\cite{WCSL10} that these $m^1$
single-photon rounds are the (only) ones contributing to the
uncertainty in terms of min-entropy about the string $X$. Formally, we have
\begin{align}
   \hmine{\eps/2}{X_0 X_1 | \Theta K} \geq \frac{m^1}{2} - \frac{m^1
  \delta}{2} \, ,
\end{align}
where $X_0, X_1$ are the sub-strings of $X$ formed according to the
index sets $\mathcal{I}_0$ and $\mathcal{I}_1$, $0<\delta < \frac14$
is fixed and the error parameter $\eps$ is
\begin{align} \label{eq:epsilon_inproof}
  \eps = 2 \exp\left( -
    \frac{(\delta/4)^2}{32(2+\log\frac{4}{\delta})^2} \cdot m^1 \right)
  \, .
\end{align}

Proceeding as in the proof of Protocol~\ref{prot:nonoise} (with $m^1$
instead of $n$), classical min-entropy splitting
(Lemma~\ref{lemma:ESLold}) then ensures that there exists a binary
random variable $D \in \set{0,1}$ such that
\[ \hmine{\eps/2}{X_{\ol{D}} |D \Theta K } \geq \frac{m^1}{4} - \frac{m^1
  \delta}{4} -1 \, .
\]

Then, we use Lemma~\ref{lem:basicminentropyincrease} to get
\begin{align*}
  \hmine{\eps}{X_{\ol{D}} | D \Theta K Q_{out}} \geq -\log
  P^{\cF}_{succ}\left( \frac{m^1}{4} - \frac{m^1 \delta}{4} -1 -\log
    \frac{2}{\eps} \right)
  \geq  -\log
  P^{\cF}_{succ}\left( \left(\frac{1}{4} - \delta \right)m^1 \right) \, ,
\end{align*}
where the last step follows in the same way as in~\cite{KWW09arxiv} from
the monotinicity of the success probability $P^{\cF}_{succ}(k) \leq
P^{\cF}_{succ}(k')$ for $k \geq k'$ and the fact that $\log \frac{2}{\eps}
\leq \frac{\delta}{2} m^1 \leq \frac{3\delta}{4} m^1-1$. 

Additionally, the dishonest receiver learns the two syndromes
$Syn(X_0), Syn(X_1)$. As $X_0$ and $X_1$ are not necessarily independent
from dishonest Bob's point of view, the two syndromes reduce Bob's
min-entropy about $X_{\ol{D}}$ by at most $1.2 \cdot
h(p^h_{\rm err}) \cdot m$ bits of information.

It follows from the chain rule for smooth min-entropy~\eqref{eq:chain} that
\begin{align*}
\hmine{\eps}{X_{\ol{D}}|D \Theta S_D Syn(X_0) Syn(X_1) K Q_{out}}
&\geq \hmine{\eps}{X_{\ol{D}} |D \Theta K Q_{out}}
- \ell - 1.2 \cdot h(p^h_{\rm err}) \cdot m \\
&\geq -\log
  P^{\cF}_{succ}\left( \left(\frac{1}{4} - \delta \right)m^1 \right) -
  \ell - 1.2 \cdot h(p^h_{\rm err}) \cdot m.
\end{align*}
The privacy amplification Theorem~\ref{thm:PA} yields
\begin{equation}\label{eq:explicitTradeoff}
  d(F_{\ol{D}}(X_{\ol{D}}) \mid D \Theta F_D S_D K Q_{out}) \leq
  2^{-\frac12( -\log P^{\cF}_{succ}\left( \left(\frac{1}{4} - \delta
      \right)m^1 \right) - 2\ell - 1.2 \cdot h(p^h_{\rm err}) \cdot m) }+ \eps
  \,
\end{equation}
which is smaller than $2 \eps$ as long as
$$
-\frac12 \log P^{\cF}_{succ}\left( \left(\frac{1}{4} - \delta \right)m^1 \right)-
\ell - 1.2 \cdot h(p^h_{\rm err}) \cdot \frac{m}{2} \geq \log\left(\frac{1}{\eps}\right) \, .
$$
from which our claim follows.
\end{proof}

In the same way as Corollary~\ref{cor:tensorproductOT}, we can derive 
\begin{corollary} \label{cor:robustOT}
  Let Bob's storage be given by $\cF=\cN^{\otimes \nu n}$
  for a storage rate~$\nu>0$, $\cN$ satisfying the strong converse
  property~\eqref{eq:strongconverseproperty} and having
  capacity~$C_\cN$ bounded by
\begin{align}
  C_\cN\cdot\nu< \left(\frac{1}{4} - \delta\right) (p^1_{\rm
      sent} - p^h_\lostB + p^d_\lostB) \, .
\end{align} 
Then Protocol~\ref{prot:practical} is $2\eps$-secure against a
dishonest receiver Bob according to Definition~\ref{def:ROT} with the
following parameters: Let $\delta\in ]0,\frac{1}{4}-C_\cN\cdot \nu[$
and $m^1 \geq 4/\delta$. Then the length $\ell$ of the OT-strings is
bounded by
\begin{align}
\ell &\leq
\frac12
\nu \cdot \gamma^\cN\left(\frac{R}{\nu}\right) \cdot n
 - 1.2 \cdot h(p^h_{\errB}) \cdot (1-p^h_\errB) \frac{n}{2}  
- \log \left( \frac{1}{\eps} \right) \ ,
\end{align}
where $\gamma^\cN$ is the strong converse
parameter of $\cN$ (see~\eqref{eq:strongconverseproperty}) and 

\bigskip
\begin{tabular}{lll}
  $m = (1-p^h_\lostB) n$ & (the number of remaining rounds) 
  \ ,\\[2mm]
  $m^1 = (p^1_\sent - p^h_\lostB + p^d_\lostB) n$ &
  (the minimal number of single-photon rounds)   
  \ ,\\[2mm]
  $R  = \left(\frac{1}{4} - \delta\right) \frac{m^1}{n}$&
  (the rate at which dishonest Bob has to send information\\
  & through storage)
  \ ,
\end{tabular}
\noindent
for sufficiently large $n$.
The error has the form
\begin{align}
\eps(\delta) \leq 2 \exp\left( 
- \frac{\delta^2}{512(4 + \log\frac{1}{\delta})^2} \cdot
(p^1_\sent - p^h_\lostB + p^d_\lostB ) n  
\right)\ .
\end{align}
\end{corollary}

\section{Password-Based Identification}\label{sec:id}
In this section, we show how the techniques for proving security in
the noisy-quantum-storage model also apply to the protocol
from~\cite{DFSS07,DFSS10journal} achieving secure password-based
identification in the bounded-quantum-storage model. This answers an
open question posed in~\cite{KWW09arxiv}.

\subsection{Task and Protocol}
A user Alice wants to identify herself to a server Bob by means of a
personal identification number (PIN). This task can be achieved by
securely evaluating the equality function on the player's inputs: Both
Alice and Bob input passwords $w_A$ and $w_B$ from a set of possible
passwords $\mathcal{W}$ into the protocol and
Bob learns as output whether $w_A = w_B$ or not.

The protocol proposed in~\cite{DFSS07} is secure against an
unbounded user Alice and a quantum-memory bounded server Bob in the
sense that it is guaranteed that if a dishonest player starts with
quantum side information which is uncorrelated with the honest
player's password $w$, this dishonest player is restricted to
guess a possible $w'$ and find out whether $w=w'$ or not while not
learning anything more than this mere bit of information about the
honest user's password $w$. Formally, security is defined as follows.
\begin{definition} \label{def:usersecurity} We call an
  identification protocol between user Alice and server Bob \emph{secure for
  the user Alice with error $\eps$} against (dishonest) server Bob
  $\dB$ if the following is satisfied: whenever the initial state of
  $\dB$ is independent of $W$, the joint state $\rho_{W \regB}$ after
  the execution of the protocol is such that there exists a random
  variable $W'$ that is independent of $W$ and such that
$$
\rho_{W W' \regB|W'\neq W} \approx_{\eps} \rho_{W\leftrightarrow W'\leftrightarrow  \regB|W'\neq W}.
$$
\end{definition}
The Markov-chain notation is explained in~\eqref{eq:markovdefinition}.

We consider the same protocol for password-based secure identification
from~\cite{DFSS07}, in the more practical form presented
in~\cite{DFLSS09}, where the receiving party measures in a random
basis. Let $\code:\mathcal{W} \rightarrow\set{+,\times}^n$ be the encoding function of
a binary code of length $n$ with $m = |\mathcal{W}|$ codewords and minimal
distance $d$. $\code$ can be chosen such that $n$ is linear in
$\log(m)$ or larger, and $d$ is linear in $n$.  Furthermore, let
${\cal F}$ and ${\cal G}$ be strongly two-universal classes of hash
functions 
from $\set{0,1}^n$ to $\set{0,1}^{\ell}$
and from $\mathcal{W}$ to $\set{0,1}^{\ell}$, respectively, for some
parameter~$\ell$. 

\begin{protocol}[\cite{DFSS07,DFLSS09}]Password-based identification \QID(w): \label{prot:id}
\begin{enumerate}
\item Alice picks $x \in_R \01^n$ and $\theta \in_R \{+,\times\}^n$.
  At time $t=0$, she sends
  $\ket{x_1}_{\theta_1},\ldots,\ket{x_n}_{\theta_n}$ to Bob.
\item Bob picks $\hat{\theta} \in_R \{+,\times\}^n$ at random and
measures the $i$th qubit in basis $\hat{\theta}_i$. 
He obtains outcome $\hat{x} \in \01^n$.
\item[]\hspace{-1cm}Both parties wait time $\Delta t$.
\item \label{step:afterwaiting}
Bob computes a string $\kappa\in \{ +, \times \}^n$ such that
$\hat{\theta}= \code(w) \oplus \kappa$ (interpreting $+$ as 0
and $\times$ as 1 so that $\oplus$ makes sense). He sends $\kappa$ to
Alice and they define the shifted code $\code'(w) \assign \code(w) \oplus \kappa$. 
\item 
Alice sends $\theta$ and $f \in_R \mathcal{F}$ to Bob. 
Both compute $\mathcal{I}_w \assign \{ i: \theta_i=\mathfrak{c'}(w)_i \}$.
\item 
Bob sends $g \in_R \mathcal{G}$ to Alice.
\item 
Alice sends $z\assign f(x|_{\mathcal{I}_w}) \oplus g(w)$ to Bob. 
\item Bob accepts if and only if $z=f(\hat{x}|_{\mathcal{I}_w}) \oplus g(w)$.
\end{enumerate}
\end{protocol}

We note that this protocol can also be (non-trivially) extended to additionally
withstand man-in-the-middle attacks~\cite{DFSS07,DFSS10journal}.

\subsection{Security Analysis}

\begin{theorem}[Security against dishonest Bob]\label{prop:Usec} 
  Fix $0<\delta<\frac14$ and let $\sigma(\delta)$ be defined as
  in~(\ref{eq:sigma}).  Then, for any attack of a dishonest Bob with
  storage channel $\cF:\cB(\cH_{in})\rightarrow\cB(\cH_{out})$,
  Protocol~\ref{prot:id} is an $\eps$-secure identification protocol
  against a dishonest receiver Bob according to
  Definition~\ref{def:usersecurity}, if $d \geq
  \frac{4+4\log(m)}{\delta}$ and
$$\eps = 2^{-\frac12( -\log P_{succ}^{\cF}\left( (\frac14 - \delta) d
  \right) - \ell)} + 2^{-(\sigma(\delta/4) d - \log(m) - 3)} \, .$$ 
\end{theorem}

%
To understand what the result on $\eps$ means,
note that using a family of asymptotically good codes, we can assume
that $d$ grows linearly with the main security parameter $n$, while
still allowing $m$ (the number of passwords) to be exponential in $n$.
So we may choose the parameters such that $\frac{d}{n},
\frac{\log(m)}{n}$, and $\frac{\ell}{n}$ are all constants. The result
above now says that $\eps$ is exponentially small as a function of $n$
if these constants and the noisy channel $\cF$ fulfill that for some
$0< \delta < \frac{1}{4}$, $\frac{ -\log P_{succ}^{\cF}\left(
    (\frac14-\delta)d \right)}{n} - \frac{\ell}{n} > 0$ and
$\sigma(\delta/4)\frac{d}{n} - \frac{\log(m)}{n}> 0$.  See Theorem
\ref{thm:impersonation} for a choice of parameters that also takes
server security into account. 


\begin{proof}

We use upper case letters $W$, $X$, $\Theta$, $K$, $F$, $G$ and $Z$ for the
random variables that describe the respective values $w$, $x$,
$\theta$ etc.\ in an execution of \QID. 

Recall that in the noisy-storage model, we denote by $K$ the classical
outcome of Bob's encoding attack and $Q_{in}$ denotes Bob's quantum state
right before the waiting time.

We write $X_j = X|_{\mathcal{I}_j}$ for any $j$. Note that dishonest
Bob starts without any knowledge about honest Alice's password $W$ and
hence, $W$ is independent of $X$, $\Theta$, $K$, $F$, $G$ and $Q_{in}$.

For $1 \leq i \neq j \leq m$, fix the value of $X$, and
correspondingly of $X_i$ and $X_j$, at the positions where $\code(i)$
and $\code(j)$ coincide, and focus on the remaining (at least) $d$
positions. The uncertainty relation (Theorem~\ref{thm:uncertainty})
implies that the restriction of $X$ to these positions has $(\frac12 -
\delta/2)d$ bits of $\eps'$-smooth min-entropy given~$\Theta$, where
$\eps' \leq 2^{-\sigma(\delta/4)d}$.
Since every bit in the restricted $X$ appears in one of
$X_i$ and $X_j$, the pair $X_i,X_j$ also has $(\frac12 - \delta/2)d$
bits of $\eps'$-smooth min-entropy given~$\Theta$ and $K$.  The Entropy
Splitting Lemma~\ref{lemma:basicES} implies that there exists $W'$
(called $V$ in Lemma~\ref{lemma:basicES}) such that if $W \neq W'$
then $X_W$ has $(\frac14 - \delta/4)d - \log(m) - 1$ bits of
$2m\eps'$-smooth min-entropy given $W$ and $W'$ (and~$\Theta,K$), i.e.,
$$ \hmin^{2m\eps'} (X_W|WW' \Theta K, W \neq W') \geq (\frac14 -
\delta/4)d - \log(m) - 1 \, .$$

By Lemma~\ref{lem:basicminentropyincrease}, it follows that for
$Q_{out}=\mathcal{F}(Q_{in})$, we get
\begin{align*} 
\hmin^{(2m+1)\eps'} (X_W|WW' \Theta K Q_{out}, W \neq W' ) &\geq
-\log P_{succ}^{\cF} \left( \left(\frac14 - \delta/4 \right) d - \log(m)
  - 1 - \log(1/\eps') \right)\\
&\geq -\log P_{succ}^{\cF} \left( \left(\frac14 - \delta \right) d
\right) \, ,
\end{align*}
where the last inequality follows as in the OT-case (proof of
Theorem~\ref{thm:secure1}) from $\log(1/\eps') \leq \frac{\delta}{2}d
\leq \frac{3\delta}{4}d - \log(m) -1$ and the assumption on $d$.

Privacy amplification then guarantees that $F(X_W)$ is $\eps''$-close
to random and independent of $F, W, W',\Theta,K$ and $Q_{out}$,
conditioned on $W \neq W'$, where \smash{$\eps'' = \frac12\cdot
  2^{-\frac12\left( -\log P_{succ}^{\cF}\left((\frac14-\delta)d\right)
      - \ell \right) } + (2m+1)\eps'$}.  It follows that $Z = F(X_W)
\oplus G(W)$ is $\eps''$-close to random and independent of $F, G, W,
W',\Theta,K$ and $Q_{out}$, conditioned on $W \neq W'$. The rest of
the argument is the same as in the original
proof~\cite{DFSS10journal}.

Formally, we want to upper bound the trace distance between $\rho_{W
  W' \regB| W'\neq W}$ and $\rho_{W\leftrightarrow W'\leftrightarrow
  \regB| W'\neq W}$.  Since the output state $\regB$ is, without loss
of generality, obtained by applying some unitary transform to the set
of registers 
$(Z,F,G,W',\Theta,K, Q_{out})$, the distance above is equal to the
distance between $\rho_{W W' (Z,F,G,\Theta,K, Q_{out})| W'\neq W}$ and
$\rho_{W\leftrightarrow W'\leftrightarrow (Z,F,G,\Theta,K, Q_{out})|
  W'\neq W}$.  We then get:
\begin{align*}
  \rho&_{W W' (Z, F, G, \Theta, Q_{out})|W'\neq W} \approx_{\eps''}
  {\textstyle\frac{1}{2^\ell}}\mathbbm{1}_Z \otimes \rho_{W W' (F, G,  \Theta,K, Q_{out})|W'\neq W} \\
  &=\; {\textstyle\frac{1}{2^\ell}}\mathbbm{1}_Z \otimes
  \rho_{W\leftrightarrow W'\leftrightarrow (F,G,\Theta,K, Q_{out})|
    W'\neq W} \approx_{\eps''} \rho_{W\leftrightarrow
    W'\leftrightarrow (Z,F,G,\Theta,K, Q_{out})| W'\neq W} \enspace ,
\end{align*}
where approximations follow from privacy amplification and the exact equality
comes from the independency of $W$, which, when conditioned 
on  $W' \neq W$, translates to independency given $W'$. 
The claim follows with $\eps = 2\eps''$ and the (crude) estimation
$2(2m+1) \leq 8m$. 
\end{proof}

\begin{theorem}[Security against dishonest Alice~\cite{DFSS07}]\label{prop:Ssec}
  If $\hmin(W) \geq 1$, then \QID is secure against dishonest user Alice
  with security error $\eps= m^2/2^\ell$.
\end{theorem}

We call an identification scheme {\em $\eps$-secure against
  impersonation attacks} if the protocol is secure for both players
with error at most $\eps$ in both cases. The following holds:

\begin{theorem}
\label{thm:impersonation}
If $\hmin(W) \geq 1$, then the identification scheme \QID (with
suitable choice of parameters) is $\eps$-secure against impersonation
attacks for any unbounded user Alice and for any server Bob with noisy storage
of the form $\cF = \cN^{\otimes \nu n}$ with $\nu >0$, where $\cN$
satisfies the strong-converse property~\eqref{eq:strongconverseproperty}, and
$$ C_{\cN} \cdot \nu < \frac14 \, ,$$
and the security error is 
$$
\eps = 2^{-\frac13( \gamma^{\cN}\left( \frac{1/4 - \delta}{\nu}
  \right) \nu \mu n - 6\log(m) - 1)} + 
2^{-(\sigma(\delta/4) \mu n - \log(m) - 4)}
$$ 
for an arbitrary $0 < \delta < \frac14$, and where $\mu =
h^{-1}(1-\log(m)/n)$, and $h^{-1}$ is the inverse function of the
binary entropy function: $h(p) \assign
-p\cdot\log(p)-(1-p)\cdot\log(1-p)$ restricted to $0 < p \leq
\frac12$.  In particular, if $\log(m)$ is sublinear in $n$, then
$\eps$ is negligible in $n$ as long as $\gamma^{\cN}\left(\frac{1/4 -
    \delta}{\nu} \right) > 0$.
\end{theorem}

\begin{proof}
First of all, we have that $-\log P_{succ}^{\cN^{\otimes \nu n}}\left(
  (1/4-\delta) d \right) \geq \gamma^{\cN}\left(
  \frac{1/4-\delta}{\nu} \right) \nu d$. 

We choose $\ell = \frac13 \cdot \gamma^{\cN} \left(
  \frac{1/4-\delta}{\nu} \right) \nu d$. Then security against
dishonest Bob holds except with an error $\eps = 2^{-\frac13 \cdot
  \gamma^{\cN} \left( \frac{1/4-\delta}{\nu} \right) \nu d} +
2^{-(\sigma(\delta/4) d - \log(m) - 3)}$, and security against
dishonest Alice holds except with an error $m^2/2^\ell = 2^{- \frac13
  \left( \gamma^{\cN} \left( \frac{1/4-\delta}{\nu} \right) \nu d - 6
    \log(m) \right)}$. Using a code $\code$, which asymptotically
meets the Gilbert-Varshamov bound~\cite{Tho83}, $d$ may be chosen
arbitrarily close to $n \cdot h^{-1}\bigl(1-\log(m)/n\bigr)$.  In
particular, we can ensure that $d$ does not differ from this value by
more than 1. Inserting $d= \mu \cdot n -1$ in the expressions and
using that $\gamma^{\mathcal{N}}\left( \frac{1/4-\delta}{\nu} \right)\nu \leq 1$ yields the theorem.
\end{proof}

\section{Conclusion}
We have used the technical tool from~\cite{KWW09arxiv} to prove the
security of the original protocols for oblivious transfer and secure
identification against adversaries performing general
noisy-quantum-storage attacks. The main advantage of our protocols
is the straightforward constant-round classical post-processing which
makes them easier to implement in the lab compared to the protocols
from~\cite{KWW09arxiv,WCSL10}. Their security analysis yields
simpler expressions for the security error. For a given number of
pulses and a low security threshold, our approach generally yields
higher OT-rates. We show for the first time the security of a
password-based identification protocol against general
noisy-quantum-storage attacks.

This work leads to the question whether a similar result as in QKD
holds, namely that general storage attacks are no better than coherent
(or individual) storage attacks for which the best encoding attack is
known~\cite{STW09}.

\section*{Acknowledgments}
This work is supported by EU fifth framework project QAP IST 015848
and the Dutch NWO VICI project 2004-2009. 

\bibliographystyle{alpha}
\bibliography{crypto,qip,procs}

\end{document}